\documentclass[review,12pt]{elsarticle}
\usepackage{times}
\usepackage{algorithmic}
\usepackage{algorithm}
\usepackage{amssymb}
\usepackage{amsmath}
\usepackage{amsthm}
\usepackage{url}
\usepackage{graphicx}
\usepackage{caption}
\usepackage{epstopdf}
\usepackage{multirow}
\usepackage{comment}
\usepackage{tikz}

\newtheorem{definition}{Definition}
\newtheorem{proposition}{Proposition}
\newtheorem{lemma}{Lemma}
\newtheorem{corollary}{Corollary}

\tikzset{
  c_node/.style={font=\tiny},
  u_node/.style={rectangle,draw,minimum size=0.3cm,inner sep=5pt},
  v_node/.style={circle,draw,minimum size=0.3cm,inner sep=3pt},
  ua_label/.style={above, sloped, inner sep=2pt, very near start},
  va_label/.style={above, sloped, inner sep=2pt, very near end},
  ub_label/.style={below, sloped, inner sep=2pt, very near start},
  vb_label/.style={below, sloped, inner sep=2pt, very near end}  
}

\begin{document}

\begin{frontmatter}
\title{An Improved Subsumption Testing Algorithm for the Optimal-Size Sorting Network Problem}

\author{Cristian Fr\u asinaru and M\u ad\u alina R\u aschip}
\address{Faculty of Computer Science, "Alexandru Ioan Cuza" University, Ia\c si, Romania}
\ead{\{acf, mionita\}@info.uaic.ro}


\begin{abstract}
In this paper a new method for checking the subsumption relation for the optimal-size sorting network problem is described.
The new approach is based on creating a bipartite graph and modelling the subsumption test as the problem of enumerating all perfect matchings in this graph. 
Experiments showed significant improvements over the previous approaches when considering the number of subsumption checks and the time needed to find optimal-size sorting networks. We were able to generate all the complete sets of filters for comparator networks with $9$ channels, 
confirming that the $25$-comparators sorting network is optimal. 
The running time was reduced more than $10$ times, compared to the state-of-the-art result described in \cite{codish:2014:9-25}.
\end{abstract}
\begin{keyword}
Comparator networks. Optimal-size sorting networks. Subsumption.
\end{keyword} 

\end{frontmatter}

\section{Introduction}
Sorting networks are a special class of sorting algorithms with an active research area since the 1950's 
\cite{knuth:1998:acp3}, \cite{batcher:1968}, \cite{baddar:2011}.
A sorting network is a comparison network which for every input sequence produces a monotonically increasing output. 
Since the sequence of comparators does not depend on the input, the network represents an oblivious sorting algorithm. 
Such networks are suitable in parallel implementations of sorting, being applied in graphics processing units \cite{kipfer:2005} 
and multiprocessor computers \cite{batcher:1968}. 

Over time, the research was focused on finding the optimal sorting networks relative to their size or depth.
When the size is considered, the network must have a minimal number of comparators, while for the second objective a minimal number of layers is required. In \cite{ajtai:1983} a construction method for sorting network of size $O(n log n)$ and depth $O(log n)$ is given. This algorithm has good results in theory but it is inefficient in practice because of the large constants hidden in the big-$O$ notation. On the other side, the simple algorithm from \cite{batcher:1968} which constructs networks of depth $O(log^2 n)$ has good results for practical values of $n$.

Because optimal sorting networks for small number of inputs can be used to construct efficient larger networks the research in the area focused in the last years on finding such small networks. Optimal-size and optimal-depth networks are known for $n \leq 8$ \cite{knuth:1998:acp3}. 
In \cite{parberry:1991:9input} the optimal-depth sorting networks were provided for $n=9$ and $n=10$. 
The results were extended for $11 \leq n \leq 16$ in \cite{bundala:2013:optimal}. The approaches use search with pruning based on symmetries on the first layers. The last results for parallel sorting networks are for 17 to 20 inputs and are given in \cite{ehlers:2015}, \cite{codish:2015:end}. On the other side, the paper \cite{codish:2014:9-25} proved the optimality in size for the case $n=9$ and $n=10$. The proof is based on exploiting symmetries in sorting networks and on encoding the problem as a satisfiability problem. 
The use of powerful modern SAT solvers to generate optimal sorting networks is also investigated in \cite{morgenstern:2011}. Other recent results can be found in \cite{codish:2014:quest}, where a revised technique to generate, modulo symmetry, the set of saturated two-layer comparator networks is given. Finding the minimum number of comparators for $n>10$ is still an open problem.  In this paper, we consider the optimal-size sorting networks problem.

Heuristic approaches were also considered in literature, for example approaches based on evolutionary algorithms \cite{valsalam:2013} 
that are able to discover new minimal networks for up to $22$ inputs, but these methods cannot prove their optimality.

One of the most important and expensive operation used in \cite{codish:2014:9-25} is the subsumption testing. This paper presents a new better approach to implement this operation based on matchings in bipartite graphs. The results show that the  new approach makes the problem more tractable by scaling it to larger inputs.

The paper is organized as follows. Section 2 describes the basic concepts needed to define the optimal-size sorting-network problem and a new model of the subsumption problem. Section 3 presents the problem of finding the minimal-size sorting network. Section 4 discusses the subsumption problem while Section 5 the subsumption testing. Section 6 presents the new way of subsumption testing by enumerating all perfect matchings. Section 7 describes the experiments made to evaluate the approach and presents the results.

\section{Basic Concepts}        

A {\it comparator network} $C_{n,k}$ with $n$ {\it channels} (also called {\it wires}) and {\it size} $k$ is a sequence of {\it comparators} 
$c_1=(i_1,j_1);\dots;c_k=(i_k; j_k)$ where each comparator $c_t$ specifies a pair of channels $1 \le i_t < j_t \le n$.
We simply denote by $C_n$ a comparator network with $n$ channels, whenever the size of the network is not significant in a certain context.

Graphically, a comparator network may be represented as a Knuth diagram \cite{knuth:1998:acp3}. 
A channel is depicted as a horizontal line and a comparator as a vertical segment connecting two channels.

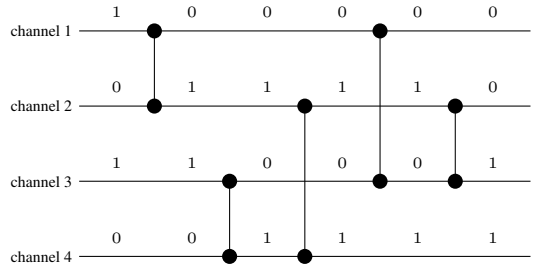
\begin{figure}[!ht]
\centering
\captionsetup{justification=centering}
\begin{tikzpicture}
\begin{scope}
   \node[c_node] at (-0.5,3) {channel 1};
   \node[c_node] at (-0.5,2) {channel 2};
   \node[c_node] at (-0.5,1) {channel 3};
   \node[c_node] at (-0.5,0) {channel 4};

   \draw (0,0) -- (6,0);
   \draw (0,1) -- (6,1);
   \draw (0,2) -- (6,2);
   \draw (0,3) -- (6,3);

   \fill (1,3) circle (1mm);
   \fill (1,2) circle (1mm);
   \draw (1,3) -- (1,2) node [font=\scriptsize, midway, left, inner sep=0pt] {}; 
   
   \fill (2,0) circle (1mm);
   \fill (2,1) circle (1mm);
   \draw (2,0) -- (2,1) node [font=\scriptsize, midway, left, inner sep=0pt] {}; 

   \fill (3,2) circle (1mm);
   \fill (3,0) circle (1mm);
   \draw (3,2) -- (3,0) node [font=\scriptsize, midway, left, inner sep=0pt] {}; 

   \fill (4,3) circle (1mm);
   \fill (4,1) circle (1mm);
   \draw (4,3) -- (4,1) node [font=\scriptsize, midway, left, inner sep=0pt] {}; 

   \fill (5,2) circle (1mm);
   \fill (5,1) circle (1mm);
   \draw (5,2) -- (5,1) node [font=\scriptsize, midway, left, inner sep=0pt] {}; 
   
   \node[c_node] at (0.5,3.25) {$1$};
   \node[c_node] at (0.5,2.25) {$0$};
   \node[c_node] at (0.5,1.25) {$1$};
   \node[c_node] at (0.5,0.25) {$0$};

   \node[c_node] at (1.5,3.25) {$0$};
   \node[c_node] at (1.5,2.25) {$1$};
   \node[c_node] at (1.5,1.25) {$1$};
   \node[c_node] at (1.5,0.25) {$0$};

   \node[c_node] at (2.5,3.25) {$0$};
   \node[c_node] at (2.5,2.25) {$1$};
   \node[c_node] at (2.5,1.25) {$0$};
   \node[c_node] at (2.5,0.25) {$1$};

   \node[c_node] at (3.5,3.25) {$0$};
   \node[c_node] at (3.5,2.25) {$1$};
   \node[c_node] at (3.5,1.25) {$0$};
   \node[c_node] at (3.5,0.25) {$1$};

   \node[c_node] at (4.5,3.25) {$0$};
   \node[c_node] at (4.5,2.25) {$1$};
   \node[c_node] at (4.5,1.25) {$0$};
   \node[c_node] at (4.5,0.25) {$1$};
   
   \node[c_node] at (5.5,3.25) {$0$};
   \node[c_node] at (5.5,2.25) {$0$};
   \node[c_node] at (5.5,1.25) {$1$};
   \node[c_node] at (5.5,0.25) {$1$};
   
\end{scope}

\end{tikzpicture}

\caption{The sorting network $C=(1,2);(3,4);(2,4);(1,3);(2,3)$, having $4$ channels and $5$ comparators, operating on the input sequence $1010$.
The output sequence is $0011$.}
\label{fig:aug-edges1}
\end{figure}

An {\it input} to a comparator network $C_n$ may be any sequence of $n$ objects taken from a totally ordered set, for instance elements in $\mathbb{Z}^n$.
Let $\overline{x}=(x_1,\dots,x_n)$ be an input sequence. Each value $x_i$ is assigned to the channel $i$ and it will "traverse" the comparator network
from left to right. Whenever the values on two channels reach a comparator $c=(i,j)$ the following happens:
if they are not in ascending order the comparator permutes the values $(x_i,x_j)$, otherwise the values will pass through the comparator unmodified.
Therefore, the {\it output} of a comparator network is always a permutation of the input. 
If $\overline{x}$ is an input sequence, we denote by $C(\overline{x})$ the output sequence of the network $C$.

A comparator network is called a {\it sorting network} if its output is sorted ascending for every possible input.

The {\it zero-one principle} \cite{knuth:1998:acp3} states that if a comparator network $C_n$ sorts correctly all $2^n$ sequences of zero and one, 
then it is a sorting network. Hence, without loss of generality, from now on we consider only comparator networks with binary input sequences.
In order to increase readability, whenever we represent a binary sequence we only write its bits; so $1010$ is actually the sequence $(1,0,1,0)$.

The {\it output set} of a comparator network is $outputs(C)=\{C(\overline{x}) | \forall \overline{x} \in \{0,1\}^n\}$.
Let $\overline{x}$ be a binary input sequence of length $n$. 
We make the following notations: 
$zeros(\overline{x})=\{1 \le i \le n | x_i = 0\}$ and 
$ones(\overline{x})=\{1 \le i \le n | x_i = 1\}$.
The output set of a comparator network $C_n$ can be partitioned into $n+1$ {\it clusters}, 
each cluster containing sequences in $outputs(C)$ having the same number of ones.
We denote by $cluster(C,p)$ the cluster containing all sequences having $p$ ones: \\
$cluster(C,p) = \{ \overline{x} \in outputs(C) \;|\; |ones(\overline{x})| = p \}$.

Consider the following simple network $C=(1,2);(3,4)$.
The output clusters of $C$ are:
$cluster(C,0)=\{0000\}$, 
$cluster(C,1)=\{0001,0100\}$, 
$cluster(C,2)=\{0011,0101,1100\}$, 
$cluster(C,3)=\{0111,1101\}$,
$cluster(C,4)=\{1111\}$.

The following proposition states some simple observations regarding the output set and its clusters.
\begin{proposition}\label{prop1}
Let $C$ be a comparator network having $n$ channels.
\begin{itemize}
\item [(a)] $C$ is the empty network $\Leftrightarrow$ $|outputs(C)|=2^n$.
\item [(b)] $C$ is a sorting network $\Leftrightarrow$ $|outputs(C)|=n+1$ (each cluster contains exactly one element). 
\item [(c)] $|cluster(C,p)| \le \binom{n}{p}$, $1 \le p \le n-1$.
\item [(d)] $|cluster(C,0)|=|cluster(C,n)|=1$.
\end{itemize}
\end{proposition}

We extend the $zeros$ and $ones$ notations to output clusters in the following manner.
Let $C$ be a comparator network. 
For all $0 \le p \le n$ we denote $zeros(C,p) = \bigcup \{zeros(\overline{x}) | \overline{x} \in cluster(C,p)\}$ and  
$ones(C,p) = \bigcup \{ones(\overline{x}) | \overline{x} \in cluster(C,p)\}$.
These sets contain all the positions between $1$ and $n$ for which there is at least one sequence in the cluster having a zero, respectively an one, 
set at that position.
Considering the clusters from the previous example, we have:
$zeros(C,0) = zeros(C,1) = zeros(C,2)=\{1,2,3,4\}$, $zeros(C,3)=\{1,3\}$, $zeros(C,4)=\emptyset$,
$ones(C,0) = \emptyset$, $ones(C,1) = \{2,4\}$, $ones(C,2)=ones(C,3)=ones(C,4)=\{1,2,3,4\}$.

We introduce the following equivalent representation of the $zeros$ and $ones$ sets, as a sequence of length $n$, 
where $n$ is the number of channels of the network, and elements taken from the set $\{0, 1\}$.
Let $\Gamma$ be a cluster:
\begin{itemize}
\item $\overline{zeros}(\Gamma) = (\gamma_1,\dots,\gamma_n)$, where $\gamma_i=0$ if $i \in zeros(\Gamma)$, otherwise  $\gamma_i=1$,
\item $\overline{ones}(\Gamma) = (\gamma'_1,\dots,\gamma'_n)$, where $\gamma'_i=1$ if $i \in ones(\Gamma)$, otherwise  $\gamma'_i=0$.
\end{itemize}
In order to increase readability, we will depict $1$ values in $\overline{zeros}$, respectively $0$ values in $\overline{ones}$ 
with the symbol $*$.
Considering again the previous example, we have:
$\overline{zeros}(C,3)=(0*0*)$ and $\overline{ones}(C,1) = (*1*1)$.

If $C$ is a comparator network on $n$ channels and $1 \le i < j \le n$ we denote by $C;(i,j)$ the {\it concatenation} of $C$ and $(i,j)$,
i.e. the network that has all the comparators of $C$ and in addition a new comparator connecting channels $i$ and $j$.
The concatenation of two networks $C$ and $C'$ having the same number of channels is denoted by $C;C'$ and 
it is defined as the sequence of all comparators in $C$ and $C'$, first the ones in $C$ and then the ones in $C'$.
In this context, $C$ represents a {\it prefix} of the network $C;C'$.
Obviously, $size(C;C')=size(C) + size(C')$.

Let $\pi$ be a permutation on $\{1,\dots,n\}$. Applying $\pi$ on a comparator network $C=(i_1,j_1);\dots;(i_k,j_k)$ will produce 
the {\it generalized} network $\pi(C)=(\pi(i_1),\pi(j_1));\dots;(\pi(i_k),\pi(j_k))$. 
It is called generalized because it may contain comparators $(i,j)$ with $i>j$, 
which does not conform to the actual definition of a standard comparator network.
An important result in the context of analyzing sorting networks (exercise 5.3.4.16 in \cite{knuth:1998:acp3}) states that 
a generalized sorting network can always be {\it untangled} such that the result is a standard sorting network of the same size.
The untangling algorithm is described in the previously mentioned exercise.
Two networks $C_a$ and $C_b$ are called {\it equivalent} if there is a permutation $\pi$ such that untangling $\pi(C_b)$ results in $C_a$.

Applying a permutation $\pi$ on a binary sequence $\overline{x}=(x_1,\dots,x_n)$ will permute the corresponding values:
$\pi(\overline{x}) = (x_{\pi(1)},\dots,x_{\pi(n)})$.
Applying $\pi$ on a set of sequences $S$ (either a cluster or the whole output set) will permute the values of all the sequences in the set:
$\pi(S) = \{\pi(\overline{x}) | \forall \overline{x} \in S \}$.
For example, consider the permutation $\pi=(4,3,2,1)$ and the set of sequences $S=\{0011,0101,1100\}$.
Then, $\pi(S)=\{1100,1010,0011\}$

\section{Optimal-size sorting networks} 
The {\it optimal size problem} regarding sorting networks is: 
"Given a positive integer $n$, what is the minimum number of comparators $s_n$ needed to create a sorting network on $n$ channels?".

Since even the problem of verifying whether a comparator network is a sorting network is known to be Co-$\mathcal{NP}$ complete \cite{parberry:1991:complexity},
we cannot expect to design an algorithm that will easily answer the optimal size problem. On the contrary.

In order to prove that $s_n \le k$, for some $k$, it is enough to find a sorting network of size $k$. 
On the other hand, to show that $s_n > k$ one should prove that no network on $n$ channels having at most $k$ comparators is a sorting network.

Let $R^{n}_{k}$ denote the set of all comparator networks having $n$ channels and $k$ comparators.
The naive approach to identify the sorting networks is by generating the whole set $R^n_k$, 
starting with the empty network and adding all possible comparators.
In order to find a sorting network on $n$ channels of size $k$, 
one could iterate through the set $R^n_k$ and inspect the output set of each network.
According to proposition \ref{prop1} (b), if the size of the output is $n+1$ then we have found a sorting network.
If no sorting network is found, we have established that $s_n > k$.

Unfortunately, the size of $R^n_k$ grows rapidly since $|R^{n}_{k}|= (n(n-1)/2)^k$ 
and constructing the whole set $R^n_k$ is impracticable even for small values of $n$ and $k$.

We are actually interested in creating a set of networks $N^n_k$ that does not include all possible networks but contains only "relevant" elements.
\begin{definition}
A \emph{complete set of filters} \cite{codish:2014:9-25} is a set $N^n_k$ of comparator networks on $n$ channels and of size $k$, 
satisfying the following properties:
\begin{itemize}
\item [(a)] If $s_n = k$ then $N^n_k$ contains at least one sorting network of size $k$.
\item [(b)] If $k < s_n = k'$ then $\exists C^{opt}_{n,k'}$ an optimal-size sorting network 
  and $\exists C_{n,k} \in N^n_k$ such that $C$ is a prefix of $C^{opt}$.
\end{itemize}
\end{definition}
Since the existence of $N^n_k$ is guaranteed by the fact that $R^n_k$ is actually a complete set of filters,
we are interested in creating such a set that is small enough (can be computed in a "reasonable" amount of time).

\section{Subsumption}        
In order to create a complete set of filters in \cite{codish:2014:9-25} it is introduced the relation of {\it subsumption}.
\begin{definition} \label{def:sub}
Let $C_a$ and $C_b$ be comparator networks on $n$ channels. 
If there exists a permutation $\pi$ on $\{1,\dots,n\}$ such that $\pi(outputs(C_a)) \subseteq outputs(C_b)$ we say that $C_a$ \emph{subsumes} $C_b$,
and we write $C_a \preceq C_b$ (or $C_a \le_{\pi} C_b$ to indicate the permutation).
\end{definition}

For example, consider the networks $C_a=(0,1);(1,2);(0,3)$ and \\ $C_b=(0,1);(0,2);(1,3)$.
Their output sets are: \\
$ouputs(C_a)=\{\{0000\},\{0001,0010\},\{0011,0110\},\{0111,1011\},\{1111\}\}$, \\
$ouputs(C_b)=\{\{0000\},\{0001,0010\},\{0011,0101\},\{0111,1011\},\{1111\}\}$. \\
It is easy to verify that $\pi=(0,1,3,2)$ has the property that $C_a \le_{\pi} C_b$.

\begin{proposition} \label{prop:same-size}
Let $C_a$ and $C_b$ be comparator networks on $n$ channels, having $|outputs(C_a)|=|outputs(C_b)|$.
Then, $C_a \preceq C_b \Leftrightarrow C_b \preceq C_a$.
\end{proposition}
\begin{proof}
Assume that $C_a \le_{\pi} C_b \Rightarrow \pi(outputs(C_a)) \subseteq outputs(C_b)$
and since $|outputs(C_a)|=|outputs(C_b)| \Rightarrow \pi(outputs(C_a)) = outputs(C_b)$. 
That means that $\pi$ is actually mapping each sequence in $outputs(C_a)$ to a distinct sequence in $outputs(C_b)$.
The inverse permutation $\pi^{-1}$ is also a mapping, this time from $outputs(C_b)$ to $outputs(C_a)$, 
implying that $\pi^{-1}(outputs(C_b)) = outputs(C_a) \Rightarrow C_b \le_{\pi^{-1}} C_a$.
\end{proof}

The following result is the key to creating a complete set of filters:
\begin{lemma}
Let $C_a$ and $C_b$ be comparator networks on $n$ channels, both having the same size, and $C_a \preceq C_b$.
Then, if there exists a sorting network $C_b;C$ of size $k$, there also exists a sorting network $C_a;C'$ of size $k$.
\end{lemma}
The proof of the lemma is presented in \cite{codish:2014:9-25} (Lemma 2) and \cite{bundala:2013:optimal} (Lemma 7).

The previous lemma "suggests" that when creating the set of networks $R^n_k$ using the naive approach, 
and having the goal of creating actually a complete set of filters, we should not add two networks in this set if one of them subsumes the other.

\vspace{5mm}
\underline{The algorithm to generate $N^n_k$}
\begin{algorithmic}
\REQUIRE $n,k \in \mathbb{Z}^+$
\ENSURE Returns $N^n_k$, a complete set of filters
\STATE $N^n_0 = \{C_{n,0}\}$
\COMMENT{Start with the empty network} 
\FORALL{$p = 1\dots k$}
  \STATE $N^n_p = \emptyset$
  \COMMENT{Generate $N^n_p$ from $N^n_{p-1}$, adding all possible comparators} 
  \FORALL{$C \in N^{n}_{p-1}$}
    \FORALL{$i = 1 \dots n-1$, $j = i+1 \dots n$}
	  \IF{the comparator $(i,j)$ is redundant}
	    \STATE {\bf continue}
	  \ENDIF
      \STATE $C^*=C;(i,j)$ 
      \COMMENT{Create a new network $C^*$}
      \IF{$\not\exists C' \in N^n_p$ such that $C' \preceq C^*$}      
		\STATE $N^n_p = N^n_p \cup C^*$		
		\STATE Remove from $N^n_p$ all the networks $C''$ such that $C^* \preceq C''$.
      \ENDIF
    \ENDFOR
  \ENDFOR
\ENDFOR
\STATE {\bf return} $N^n_k$
\end{algorithmic}

A comparator $c$ is {\it redundant} relative to the network $C$ if adding it at the end of $C$ does not modify the output set:
$outputs(C;c) = outputs(C)$.
Testing if a comparator $c=(i,j)$ is redundant relative to a network $C$ can be easily implemented by inspecting the values $x_i$ and $x_j$ 
in all the sequences $\overline{x} \in outputs(C)$. If $x_i \le x_j$ for all the sequences then $c$ is redundant.

The key aspect in implementing the algorithm above is the test for subsumption.

\section{Subsumption testing}
Let $C_a$ and $C_b$ be comparator networks on $n$ channels. According to definition \ref{def:sub}, in order to check if $C_a$ subsumes $C_b$ 
we must find a permutation $\pi$ on $\{1,\dots,n\}$ such that $\pi(outputs(C_a)) \subseteq outputs(C_b)$.
If no such permutation exists then $C_a$ does not subsume $C_b$.

In order to avoid iterating through all $n!$ permutations, 
in \cite{codish:2014:9-25} several results are presented that identify situations when subsumption testing can be implemented efficiently. 
We enumerate them as the tests $ST_1$ to $ST_4$.

{\bf $(ST_1)$ Check the total size of the output} \\
If $|outputs(C_a)| > |outputs(C_b)|$ then $C_a$ cannot subsume $C_b$.

{\bf $(ST_2)$ Check the size of corresponding clusters} (Lemma 4 in \cite{codish:2014:9-25}) \\
If there exists $0\le p \le n$ such that $|cluster(C_a,p)| > |cluster(C_b, p)|$ then $C_a$ cannot subsume $C_b$. 
When applying a permutation $\pi$ on a sequence in $outputs(C_a)$, the number of bits set to $1$ remains the same, only their positions change.
So, if $\pi(outputs(C_a))\subseteq outputs(C_b)$ then $\forall 0\le p \le n\:$ $\pi(cluster(C_a),p) \subseteq cluster(C_b,p)$, 
which implies that $|cluster(C_a)| = |\pi(cluster(C_a),p)| \le$ \\ $|cluster(C_b,p)|$  for all $0\le p \le n$.

{\bf $(ST_3)$ Check the ones and zeros} (Lemma 5 in \cite{codish:2014:9-25})\\
Recall that $zeros$ and $ones$ represent the sets of positions that are set to $0$, respectively to $1$.
If there exists $0\le p \le n$ such that $|zeros(C_a,p)| > |zeros(C_b,p)|$ or $|ones(C_a,p)| > |ones(C_b,p)|$ then $C_a$ cannot subsume $C_b$. \\
For example, consider the networks $C_a=(0,1);(2,3);(1,3);(0,4);(0,2)$ and $C_b=(0,1);(2,3);(0,2);(2,4);(0,2)$.
$cluster(C_a, 2)=\{0011,00110,01010\}$, $cluster(C_b, 2)=\{00011,01001,01010\}$,
$ones(C_a, 2)=\{2,3,4,5\}$, $ones(C_b, 2)=\{2,4,5\}$, therefore $C_a \not\preceq C_b$.

{\bf $(ST_4)$ Check all permutations} (Lemma 6 in \cite{codish:2014:9-25}) \\
The final optimization presented in \cite{codish:2014:9-25} states that 
if there exists a permutation $\pi$ such that $\pi(outputs(C_a)) \subseteq outputs(C_b)$ then
$\forall 0\le p \le n$ $zeros(\pi(C_a,p)) \subseteq zeros(C_b,p)$ and $ones(\pi(C_a, p)) \subseteq ones(C_b, p)$. 
So, before checking the inclusion for the whole output sets, we should check the inclusion for the $zeros$ and $ones$ sets, 
which is computationally cheaper.

The tests $(ST_1)$ to $(ST_3)$ are very easy to check and are highly effective in reducing the search space.
However, if none of them can be applied, we have to enumerate the whole set of $n!$ permutations, verify $(ST_4)$ and eventually
the definition of subsumption, for each one of them.
In \cite{codish:2014:9-25} the authors focused on $n=9$ which means verifying $362,880$ permutations for each subsumption test.
They were successful in creating all sets of complete filters $N^9_k$ for $k=1,\dots,25$ and actually proved that $s_9=25$.
Using a powerful computer and running a parallel implementation of the algorithm on $288$ threads, 
the time necessary for creating these sets was measured in days (more than five days only for $N^9_{14}$).

Moving from $9!$ to $10!=3,628,800$ or $11!=39,916,800$ does not seem feasible.
We have to take in consideration also the size of the complete filter sets, for example $|N^9_{14}|=914,444$.

We present a new approach for testing subsumption, which greatly reduces the number of permutations which must be taken into consideration.
Instead of enumerating all permutations we will enumerate all perfect matchings in a bipartite graph created for the networks $C_a$ and $C_b$ being tested.

\section{Enumerating perfect matchings} 
\begin{definition}
Let $C_a$ and $C_b$ be comparator networks on $n$ channels.
The \emph{subsumption graph} $G(C_a,C_b)$ is defined as the bipartite graph $(A, B; E(G))$ 
with vertex set $V(G)=A \cup B$,  where $A=B=\{1,\dots,n\}$ and the edge set $E(G)$ defined as follows.
Any edge $e\in E(G)$ is a 2-set $e=\{i,j\}$ with $i\in A$ and $j\in B$ (also written as $e=ij$) having the properties:
\begin{itemize}
\item $i \in zeros(C_a,p) \Rightarrow j \in zeros(C_b,p)$, $\forall 0 \le p \le n$;
\item $i \in ones(C_a,p)  \Rightarrow j \in ones(C_b,p)$,  $\forall 0 \le p \le n$.
\end{itemize}
\end{definition}
So, the edges of the subsumption graph $G$ represent a relationship between positions in the two output sets of $C_a$ and $C_b$. 
An edge $ij$ signifies that the position $i$ (regarding the sequences in $outputs(C_a)$) and the position $j$ (regarding $C_b$) 
are "compatible", meaning that a permutation $\pi$ with the property $\pi(outputs(C_a)) \subseteq outputs(C_b)$ might have the mapping $i$ to $j$ as a part of it.


As an example, consider the following $zeros$ and $ones$ sequences, 
corresponding to $C_a=(0,1);(2,3);(1,3);(1,4)$ and $C_b=(0,1);(2,3);(0,3);(1,4)$.
$\overline{zeros}(C_a)=\{${\tt 00000,00000,000-0,000--,000--,-----}$\}$,\\
$\overline{zeros}(C_b)=\{${\tt 00000,00000,00000,000--,000--,-----}$\}$,\\
$\overline{ones}(C_a)=\{${\tt -----,---11,1-111,11111,11111,11111}$\}$,\\
$\overline{ones}(C_b)=\{${\tt -----,---11,-1111,11111,11111,11111}$\}$.

The subsumption graph $G(C_a,C_b)$ is pictured below:
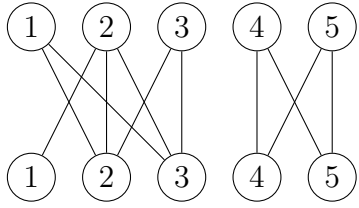
\begin{figure}[!ht]
\centering
\captionsetup{justification=centering}
\begin{tikzpicture}
\begin{scope}
   \node[v_node] (a1) at (-2,1) {$1$};
   \node[v_node] (a2) at (-1,1) {$2$};
   \node[v_node] (a3) at (0,1) {$3$};
   \node[v_node] (a4) at (1,1) {$4$};
   \node[v_node] (a5) at (2,1) {$5$};

   \node[v_node] (b1) at (-2,-1) {$1$};
   \node[v_node] (b2) at (-1,-1) {$2$};
   \node[v_node] (b3) at (0,-1) {$3$};
   \node[v_node] (b4) at (1,-1) {$4$};
   \node[v_node] (b5) at (2,-1) {$5$};
   
 	\draw[font=\tiny] (a1)--(b2);
 	\draw[font=\tiny] (a1)--(b3);
 	\draw[font=\tiny] (a2)--(b1);
 	\draw[font=\tiny] (a2)--(b2);
 	\draw[font=\tiny] (a2)--(b3);
 	\draw[font=\tiny] (a3)--(b2);
 	\draw[font=\tiny] (a3)--(b3);
 	\draw[font=\tiny] (a4)--(b4);
 	\draw[font=\tiny] (a4)--(b5);
 	\draw[font=\tiny] (a5)--(b4);
 	\draw[font=\tiny] (a5)--(b5);
\end{scope}
\end{tikzpicture}
\caption{The subsumption graph corresponding to the comparator networks $C_a=(0,1);(2,3);(1,3);(1,4)$ and $C_b=(0,1);(2,3);(0,3);(1,4)$}
\label{fig:graph}
\end{figure}

A {\it matching} $M$ in the graph  $G$ is a set of independent edges (no two edges in the matching share a common node).
If $ij \in M $ we say that $i$ and $j$ are {\it saturated}. 
A {\it perfect matching} is a matching that saturates all vertices of the graph.

\begin{lemma}
Let $C_a$ and $C_b$ be comparator networks on $n$ channels. 
If $C_a \le_{\pi} C_b$ then $\pi$ represents a perfect matching in the subsumption graph $G(C_a,C_b)$.
\end{lemma}
\begin{proof}
Suppose that $C_a \le_{\pi} C_b$, $\pi(i)=j$ and $ij \not\in E(G)$. 
That means that $\exists 0\le p \le n$ such that $i \in zeros(C_a,p) \wedge j \not\in zeros(C_b,p)$
or $i \in ones(C_a,p) \wedge j \not\in ones(C_b,p)$. We will asumme the first case.
Let $\overline{x}$ a sequence in $cluster(C_a,p)$ such that $\overline{x}(i)=0$. 
Since $\pi(outputs(C_a)) \subseteq outputs(C_b) \Rightarrow \pi(\overline{x}) \in cluster(C_b, p)$.
But $\pi(i)=j$, therefore in $cluster(C_b, p)$ there is the sequence $\pi(\overline{x})$ having the bit at position $j$ equal to $0$,
contradiction.
\end{proof}

The previous lemma leads to the following result:
\begin{corollary} \label{cor}
Let $C_a$ and $C_b$ be comparator networks on $n$ channels.
Then $C_a$ subsumes $C_b$ if and only if there exists a perfect matching $\pi$ in the subsumption graph $G(C_a,C_b)$.
\end{corollary}

The graph in figure \ref{fig:graph} has only four perfect matchings: $(2,1,3,4,5)$, $(3,1,2,4,5)$, $(2,1,3,5,4)$, $(3,1,2,5,4)$.
So, when testing subsumption, instead of verifying $5!=120$ permutations it is enough to verify only $4$ of them.

If two clusters are of the same size, then we can strengthen the previous result even more.
If there is a permutation $\pi$ such that $\pi(cluster(C_a, p)) = cluster(C_b, p)$ then $\pi^{-1}(cluster(C_b, p) = cluster(C_a, p)$.
Using the same reasoning, when creating the subsumption graph $C(G_a,C_b)$ we add the following two condition when defining an edge $ij$:
\begin{itemize}
\item $j \in zeros(C_b,p) \Rightarrow i \in zeros(C_a,p)$, $\forall 0 \le p \le n$ such that $|cluster(C_a,p)|=|cluster(C_b,p)|$,
\item $j \in ones(C_b,p) \Rightarrow i \in ones(C_a,p)$, $\forall 0 \le p \le n$ such that $|cluster(C_a,p)|=|cluster(C_b,p)|$.
\end{itemize}

In order to enumerate all perfect matchings in a bipartite graph, we have implemented the algorithm described in \cite{uno:1997}.
The algorithm starts with finding a perfect matching in the subsumption graph $G(C_a,C_b)$. 
Taking into consideration the small size of the bipartite graph, we have chosen the Ford-Fulkerson algorithm which is very simple and does not
require elaborate data structures. Its time complexity is $O(n|E(G)|)$.
If no perfect matching exists, then we have established that $C_a$ does not subsume $C_b$.
Otherwise, the algorithm presented in \cite{uno:1997} identifies all other perfect matchings, taking only $O(n)$ time per matching.

\section{Experimental results} 

We implemented both variants of subsumption testing: 
\begin{itemize}
\item $(1)$ enumerating all permutations and checking the inclusions described by $(ST_4)$ before verifying the actual definition of subsumption;
\item $(2)$ verifying only the permutations that are actually perfect matchings in the subsumption graph, according to Corollary \ref{cor}.
\end{itemize}


We made some simple experiments on a regular computer (Intel i7-4700HQ @2.40GHz), using $8$ concurrent threads.
The programming platform was Java SE Development Kit 8.

Several suggestive results are presented in the table below:
{\footnotesize
\begin{center}
\begin{tabular}[t]{|l|r|r|r|r|r|r|r|}
\hline
$(n,k)$				&	$|N^n_k|$	& $total$		& $sub$		& $perm_1$			&	$time_1$	&	$perm_2$	&	$time_2$	\\ \hline
$(7,9)$  			& 	$678$		& $1,223,426$	& $5,144$	& $26,505,101$		&	$2.88$		&	$33,120$	&	$0.07$		\\ \hline
$(7,10)$  			& 	$510$		& $878,995$		& $5,728$	& $25,363,033$		&	$2.82$		&	$24,362$	&	$0.06$		\\ \hline
$(8,7)$  			& 	$648$		& $980,765$		& $2,939$	& $105,863,506$		&	$13.67$		&	$49,142$	&	$0.14$		\\ \hline
$(8,8)$  			& 	$2088$		& $9,117,107$	& $9,381$	& $738,053,686$		&	$94.50$		&	$283,614$	&	$0.49$		\\ \hline
$(8,9)$  			& 	$5703$		& $24,511,628$	& $29,104$	& $4,974,612,498$	&	$650.22$	&	$1,303,340$	&	$1.96$		\\ \hline
\end{tabular}
\end{center}
}

The columns of the table have the following significations:
\begin{itemize}
\item $(n,k)$ - $n$ is the number of channels, $k$ is the number of comparators;
\item $|N^n_k|$ - the size of the complete set of filters generated for the given $n$ and $k$;
\item $total$ - the total number of subsumption checks;
\item $sub$ - the number of subsumptions that were identified;
\item $perm_1$ - how many permutations were checked, using the variant $(1)$;
\item $time_1$ - the total time, measured in seconds, using the variant $(1)$;
\item $perm_2$ - how many permutations were checked, using the variant $(2)$;
\item $time_2$ - the total time, measured in seconds, using the variant $(2)$;
\end{itemize}

As we can see from this results, using the variant $(2)$ the number of permutations that were verified in order to establish subsumption is greatly reduced.
Despite the fact that it is necessary to create the subsumption graph and to iterate through its set of perfect matchings, 
this leads to a much shorter time needed for the overall generation of the complete set of filters.

This new approach enabled us to reproduce the state-of-the-art result concerning optimal-size sorting networks, described in \cite{codish:2014:9-25}.
Using an Intel Xeon E5-2670 @ 2.60GHz computer, with a total of 32 cores,
we generated all the complete set of filters for $n=9$. The results are presented in the table below.

{\footnotesize
\begin{center}
\begin{tabular}[t]{|l|r|r|r|r|r|r|r|r|r|}
\hline
$k$			&	$1$	& $2$	& $3$	& $4$	& $5$	& $6$	& $7$	& $8$		&	\\ \hline
$|N^{9}_k|$	&	$1$	& $3$	& $7$	& $20$	& $59$	& $208$	& $807$	& $3415$	&	\\ \hline
$time(s)$	&	$0$	& $0$	& $0$	& $0$	& $0$	& $0$	& $0$	& $0$		&	\\ \hline \hline
$k$			&	$9$		& $10$		& $11$		& $12$		& $13$		&	$14$	& $15$		& $16$		& \\ \hline 
$|N^{9}_k|$	&	$14343$	& $55991$	& $188730$	& $490322$	& $854638$	&	$914444$& $607164$	& $274212$	& \\ \hline
$time(s)$	&	$4$		& $48$		& $769$		& $6688$	& $25186$	&	$40896$	& $24161$	& $5511$	& \\ \hline \hline
$k$			&	$17$	& $18$		& $19$		& $20$		& $21$	& $22$	& $23$	& $24$	& $25$ 	\\ \hline
$|N^{9}_k|$	&	$94085$	& $25786$	& $5699$	& $1107$	& $250$	& $73$	& $27$	& $8$	& $1$ 	\\ \hline
$time(s)$	&	$610$	& $36	$	& $2$		& $0$		& $0$	& $0$	& $0$	& $0$	& $0$ 	\\ \hline
\end{tabular}
\end{center}
}

In \cite{codish:2014:9-25} the necessary time required to compute $|N^{9}_{14}|$ using the generate-and-prune approach 
was estimated at more than $5$ days of computation on $288$ threads. 
Their tests were performed on a cluster with a total of 144 Intel E8400 cores clocked at 3 GHz.
In our experiments, the same set was created in only $11$ hours, which is actually a significant improvement.

\section{Acknowledgments}
We would like to thank Michael Codish for introducing us to this research topic and Cornelius Croitoru for his valuable comments. 
Furthermore, we thank Mihai Rotaru for providing us with the computational resources to run our experiments.

\section{Conclusions}
In this paper we have extended the work in \cite{codish:2014:9-25}, further investigating the relation of subsumption.
In order to determine the minimal number of comparators needed to sort any input of a given length, a systematic BFS-like algorithm
generates incrementally complete sets of filters, that is sets of comparator networks that have the potential to prefix an optimal-size sorting network.
To make this approach feasible it is essential to avoid adding into these sets networks that subsume one another.
Testing the subsumption is an expensive operation, invoked a huge number of times during the execution of the algorithm.
We described a new approach to implement this test, based on enumerating perfect matchings in a bipartite graph, called the subsumption graph.
Computer experiments have shown significant improvements, greatly reducing the number of invocations and the overall running time.
The results show that, using appropriate hardware, it might be possible to approach in this manner the optimal-size problem 
for sorting networks with more than $10$ channels.

\bibliographystyle{plain}
\bibliography{sorting-networks}

\end{document}